\documentclass{article}
\usepackage[utf8]{inputenc}

\usepackage[T1]{fontenc}
\usepackage{amsmath,amssymb,amsfonts,mathrsfs,bm}
\usepackage{mathtools}
\usepackage{amsthm}
\usepackage{nicefrac}
\usepackage{cite}
\usepackage[shortlabels]{enumitem}
\usepackage{graphicx}
\usepackage{epstopdf}
\usepackage{url}
\usepackage{colortbl}
\usepackage{booktabs}
\usepackage{multirow}
\usepackage[table,dvipsnames]{xcolor}
\usepackage[normalem]{ulem}

\usepackage{array}
\newcolumntype{L}[1]{>{\raggedright\let\newline\\\arraybackslash\hspace{0pt}}m{#1}}
\newcolumntype{C}[1]{>{\centering\let\newline\\\arraybackslash\hspace{0pt}}m{#1}}
\newcolumntype{R}[1]{>{\raggedleft\let\newline\\\arraybackslash\hspace{0pt}}m{#1}}

\makeatletter
\let\MYcaption\@makecaption
\makeatother
\usepackage[font=footnotesize]{subcaption}
\makeatletter
\let\@makecaption\MYcaption
\makeatother

\usepackage{xparse}



\usepackage{glossaries}

\makeatletter
\let\oldgls\gls
\let\oldglspl\glspl

\newcommand\fussy@ifnextchar[3]{%
  \let\reserved@d=#1%
  \def\reserved@a{#2}%
  \def\reserved@b{#3}%
  \futurelet\@let@token\fussy@ifnch}
\def\fussy@ifnch{%
  \ifx\@let@token\reserved@d
    \let\reserved@c\reserved@a 
  \else
    \let\reserved@c\reserved@b
  \fi
 \reserved@c}

\renewcommand{\gls}[1]{%
  \oldgls{#1}\fussy@ifnextchar.{\@checkperiod}{\@}}
\renewcommand{\glspl}[1]{%
  \oldglspl{#1}\fussy@ifnextchar.{\@checkperiod}{\@}}

\newcommand{\@checkperiod}[1]{%
  \ifnum\sfcode`\.=\spacefactor\else#1\fi
}
\makeatother

\newacronym{wrt}{w.r.t.}{with respect to}
\newacronym{RHS}{RHS}{right-hand side}
\newacronym{LHS}{LHS}{left-hand side}
\newacronym{iid}{i.i.d.}{independent and identically distributed}

\usepackage{float}

\ifx\notloadhyperref\undefined
	\ifx\loadbibentry\undefined
		\usepackage[hidelinks,hypertexnames=false]{hyperref} 
	\else
		\usepackage{bibentry}
		\makeatletter\let\saved@bibitem\@bibitem\makeatother
		\usepackage[hidelinks,hypertexnames=false]{hyperref}
		\makeatletter\let\@bibitem\saved@bibitem\makeatother
	\fi
\else
	\ifx\loadbibentry\undefined
		\relax
	\else
		\usepackage{bibentry}
	\fi
\fi

\usepackage[capitalize]{cleveref}
\crefname{equation}{}{}
\Crefname{equation}{}{}
\crefname{claim}{claim}{claims}
\crefname{step}{step}{steps}
\crefname{line}{line}{lines}
\crefname{condition}{condition}{conditions}
\crefname{dmath}{}{}
\crefname{dseries}{}{}
\crefname{dgroup}{}{}

\crefname{Problem}{Problem}{Problems}
\crefformat{Problem}{Problem~(#2#1#3)}
\crefrangeformat{Problem}{Problems~(#3#1#4) to~(#5#2#6)}

\crefname{Theorem}{Theorem}{Theorems}
\crefname{Corollary}{Corollary}{Corollaries}
\crefname{Proposition}{Proposition}{Propositions}
\crefname{Lemma}{Lemma}{Lemmas}
\crefname{Definition}{Definition}{Definitions}
\crefname{Example}{Example}{Examples}
\crefname{Assumption}{Assumption}{Assumptions}
\crefname{Remark}{Remark}{Remarks}
\crefname{Rem}{Remark}{Remarks}
\crefname{remarks}{Remarks}{Remarks}
\crefname{Appendix}{Appendix}{Appendices}
\crefname{Exercise}{Exercise}{Exercises}
\crefname{Theorem_A}{Theorem}{Theorems}
\crefname{Corollary_A}{Corollary}{Corollaries}
\crefname{Proposition_A}{Proposition}{Propositions}
\crefname{Lemma_A}{Lemma}{Lemmas}
\crefname{Definition_A}{Definition}{Definitions}

\usepackage{crossreftools}
\ifx\notloadhyperref\undefined
\else
	\relax
\fi

\usepackage{algorithm,algorithmic}

\ifx\loadbreqn\undefined
	\relax
\else
	\usepackage{breqn} 
\fi

\ifx\renewtheorem\undefined
\ifx\useTheoremCounter\undefined
\newtheorem{Theorem}{Theorem}
\newtheorem{Corollary}{Corollary}
\newtheorem{Proposition}{Proposition}
\newtheorem{Lemma}{Lemma}
\else
\newtheorem{Theorem}{Theorem}
\newtheorem{Corollary}[theorem]{Corollary}

\fi

\newtheorem{Definition}{Definition}

\newtheorem{Assumption}{Assumption}

\fi

\theoremstyle{remark}

\theoremstyle{plain}

\newcommand{\calC}{\mathcal{C}}

\newcommand{\calF}{\mathcal{F}}
\newcommand{\calG}{\mathcal{G}}

\newcommand{\calM}{\mathcal{M}}

\newcommand{\calO}{\mathcal{O}}
\newcommand{\calP}{\mathcal{P}}

\newcommand{\bP}{\mathbf{P}}

\DeclareSymbolFont{bsfletters}{OT1}{cmss}{bx}{n}
\DeclareSymbolFont{ssfletters}{OT1}{cmss}{m}{n}
\DeclareMathSymbol{\bsfGamma}{0}{bsfletters}{'000}
\DeclareMathSymbol{\ssfGamma}{0}{ssfletters}{'000}
\DeclareMathSymbol{\bsfDelta}{0}{bsfletters}{'001}
\DeclareMathSymbol{\ssfDelta}{0}{ssfletters}{'001}
\DeclareMathSymbol{\bsfTheta}{0}{bsfletters}{'002}
\DeclareMathSymbol{\ssfTheta}{0}{ssfletters}{'002}
\DeclareMathSymbol{\bsfLambda}{0}{bsfletters}{'003}
\DeclareMathSymbol{\ssfLambda}{0}{ssfletters}{'003}
\DeclareMathSymbol{\bsfXi}{0}{bsfletters}{'004}
\DeclareMathSymbol{\ssfXi}{0}{ssfletters}{'004}
\DeclareMathSymbol{\bsfPi}{0}{bsfletters}{'005}
\DeclareMathSymbol{\ssfPi}{0}{ssfletters}{'005}
\DeclareMathSymbol{\bsfSigma}{0}{bsfletters}{'006}
\DeclareMathSymbol{\ssfSigma}{0}{ssfletters}{'006}
\DeclareMathSymbol{\bsfUpsilon}{0}{bsfletters}{'007}
\DeclareMathSymbol{\ssfUpsilon}{0}{ssfletters}{'007}
\DeclareMathSymbol{\bsfPhi}{0}{bsfletters}{'010}
\DeclareMathSymbol{\ssfPhi}{0}{ssfletters}{'010}
\DeclareMathSymbol{\bsfPsi}{0}{bsfletters}{'011}
\DeclareMathSymbol{\ssfPsi}{0}{ssfletters}{'011}
\DeclareMathSymbol{\bsfOmega}{0}{bsfletters}{'012}
\DeclareMathSymbol{\ssfOmega}{0}{ssfletters}{'012}

\newcommand{\qednew}{\nobreak \ifvmode \relax \else
      \ifdim\lastskip<1.5em \hskip-\lastskip
      \hskip1.5em plus0em minus0.5em \fi \nobreak
      \vrule height0.75em width0.5em depth0.25em\fi}

\newcommand{\ud}{\mathrm{d}}

\newcommand{\norm}[1]{{\left\lVert{#1}\right\rVert}}

\DeclareDocumentCommand\set{s m t| m}{%
  \IfBooleanTF#1%
	{\left\{\, #2\mathrel{} \IfBooleanTF{#3}{\middle|}{:}\mathrel{}  #4\, \right\}}%
  {\{\, #2 \IfBooleanTF{#3}{\mid}{\mathrel{} : \mathrel{}} #4\, \}}%
}

\DeclareDocumentCommand \ifcond {m m} {%
	{#1} %
	\IfValueT{#2}{\, \middle|\, {#2}}%
}

\DeclareDocumentCommand \P {e{_} g >{\SplitArgument{ 1 }{ @| }}d() g } {%
	\mathbb{P}%
	\IfValueTF{#1}{_{#1}}
		{\IfValueT{#2}{_{#2}}}%
	\IfValueT{#3}{\left(\ifcond#3}%
	\IfValueT{#4}{\, \middle|\, {#4}}%
	\IfValueT{#3}{\right)}%
}

\DeclareDocumentCommand \E {e{_} g >{\SplitArgument{ 1 }{ @| }}o g } {%
	\mathbb{E}%
	\IfValueTF{#1}{_{#1}}
		{\IfValueT{#2}{_{#2}}}%
	\IfValueT{#3}{\left[\ifcond#3}%
	\IfValueT{#4}{\, \middle|\, {#4}}%
	\IfValueT{#3}{\right]}%
}

\definecolor{gray90}{gray}{0.9}

\ifx\nohighlights\undefined

	\newcommand{\msout}[1]{\text{\color{green} \sout{\ensuremath{#1}}}}
	\newcommand{\del}[1]{{\color{green}\ifmmode \msout{#1}\else\sout{#1}\fi}}
\else

	\newcommand{\msout}[1]{#1}
	\newcommand{\del}[1]{#1}
\fi

\newcommand{\hhide}[1]{}

\renewcommand{\figurename}{Fig.}
\newcommand{\figref}[1]{\figurename~\ref{#1}}
\graphicspath{{./Figures/}} 
\newcommand{\includeCroppedPdf}[2][]{%
    \IfFileExists{./Figures/#2-crop.pdf}{}{%
        \immediate\write18{pdfcrop ./Figures/#2 ./Figures/#2-crop.pdf}}%
    \includegraphics[#1]{./Figures/#2-crop.pdf}}

\ifx\diagnoselabel\undefined
	\relax
\else
	\makeatletter
	 \def\@testdef #1#2#3{%
		 \def\reserved@a{#3}\expandafter \ifx \csname #1@#2\endcsname
		\reserved@a  \else
	 \typeout{^^Jlabel #2 changed:^^J%
	 \meaning\reserved@a^^J%
	 \expandafter\meaning\csname #1@#2\endcsname^^J}%
	 \@tempswatrue \fi}
	\makeatother
\fi

\usepackage{tikz-cd} 

\title{Graph signal processing with categorical perspective}
\author{Feng Ji, Xingchao Jian, Wee Peng Tay\footnote{The authors are with the School of Electrical and Electronic Engineering, Nanyang Technological University, 639798, Singapore (e-mail: jifeng@ntu.edu.sg, xingchao001@e.ntu.edu.sg, wptay@ntu.edu.sg).}}

\begin{document}

\maketitle

\section{Introduction} \label{sec:intro}
Graph signal processing (GSP) \cite{Shu13} has emerged as a powerful framework for analyzing graph-structured data, where the central concept is the vector space of graph signals. For a graph $G$ of size $n$, the vector space $\mathbb{R}^n$ of graph signals has enabled the development of a variety of signal processing tools, which have been successfully applied in various domains, such as social network analysis, the study of sensor networks, and transportation network analysis. However, uncertainty is omnipresent in practice, and using a vector to model a real signal can be erroneous in some situations.

To address this challenge, we propose (cf.\ \cite{Ji23d}) to use Wasserstein space \cite{Vil09} as a replacement for the vector space of graph signals, to account for signal stochasticity. The Wasserstein space is strictly more general than the classical graph signal space and provides a more flexible and realistic framework for modeling uncertain signals. An element in the Wasserstein space is called a distributional graph signal in \cite{Ji23d}.

On the other hand, signal processing for a probability space of graphs and operators has also been proposed in \cite{Ji22}. It is integrated with the notion of distributional graph signals in \cite{Ji23d}, and the unified framework that encompasses existing theories regarding graph uncertainty and provides a more comprehensive approach to analyzing graph-structured data. 

The approach is concrete and follows the observation that distributional graph signal transformation can be built up from two types of measurable functions $p: \mathbb{R}^n \times \calG_n \to \mathbb{R}^n$ and $f: \mathbb{R}^n \times \calG_n \to \mathbb{R}^n$, where $\calG_n$ the space of simple weighted graphs on $n$ vertices. The function $p$ is the projection and $f$ factors as $\mathbb{R}^n \times \calG_n \to \mathbb{R}^n\times M_n(\mathbb{R}), (r,G)\mapsto (r,h(G))$ and $\mathbb{R}^n\times M_n(\mathbb{R}) \to \mathbb{R}^n, (r,M) = Mr$, for some measurable function $h: \calG_n \to M_n(\mathbb{R})$. The setup is still restrictive. For example, the size $n$ is used throughout, and hence the change of graph such as graph augmentation is not included as a special case.  

In this paper, we propose a vast generalization for graph signal processing with a categorical perspective. Recall that the category theory \cite{Hun03} deals with the abstract study of structures and relationships between objects. It provides a framework for organizing mathematical concepts and objects. In our case, we shall use categorical language to formalize the notion of signal adaptive graph structures (SAGS) and associated filters introduced in \cite{Ji23d}, as morphisms of the category of correspondences $\calC_{\mathfrak{c}}$ in \cref{sec:tco}. The aim of the paper is to construct subcategories of $\calC_{\mathfrak{c}}$ that have important signal processing significance. We can use the framework to give an abstract and unified perspective of many concrete GSP ideas.  

\section{Abstract graph signals}
As we have mentioned in \cref{sec:intro}, we want to unify the two directions of generalization of the traditional GSP. Recall that if $X$ is a metric space, define $\mathcal{P}(X)$ to be the space of probability measures on $X$ with finite mean and variance, the \emph{$2$-Wasserstein space} on $X$ \cite{Vil09}. In traditional GSP, the signal space is $\mathbb{R}^n$. It embeds isometrically in $\mathcal{P}(\mathbb{R}^n)$ via $r\mapsto \delta_r$, where $\delta_r$ is the delta distribution at $r$. On the other hand, \cite{Ji22} considers a distribution $\mu$ of graph shift operators. It can be interpreted (trivially) as for each graph signal $r$, a distribution of $\mu_r = \mu$ of operators is associated with $r$. To generalize, we consider fiberwise measure as follows.

\begin{Definition}
For measurable spaces $X,Y$, suppose $f:Y \to X$ is \emph{surjective and measurable}. The map $f$ is said to be equipped with \emph{fiberwise probability measure} if for each $x \in X$, there is a probability measure on $\mu_x$ on the fiber $f^{-1}(x)$. For convenience, we may say such an $f$ is an \emph{fpm}. We call $Y$ the \emph{total space} and $X$ the \emph{base space}. 
\end{Definition}

If $f$ has fiberwise probability measure, then for each probability measure $\mu$ on $X$, there is a pullback measure $f^*(\mu)$ on $Y$ defined by \begin{align*} f^*(\mu)(U) = \int_{x\in X} \mu_x\Big(U\cap f^{-1}(x)\Big) d\mu, \end{align*}
for measurable subset $U$ of $Y$. We shall use $f$ with fiberwise probability measure for base space $\mathbb{R}^n$ to model joint information of graphs and signals. For example, if $\Gamma_n$ is the (discrete) space of unweighted graphs on $n$ vertice, then an fpm $p: \mathbb{R}^n\times \Gamma_n \to \mathbb{R}^n, (r,G) \mapsto r$ associates a distribution of graphs $\mu_r$ (on $\Gamma_n$) for each signal $r$.  

As a side remark, in the reverse direction, for measurable $g: X \to Y$ and a probability $\mu$ on $X$, there is a pushforward (or induced) measure $g_*(\mu)$ on $Y$. It is defined by $g_*(\mu)(U) = \mu(g^{-1}(U))$ for any measurable subset $U$ of $Y$. 

Suppose we are given an fpm $f: Y \to X$. A measurable function $g: Y \to Y$ induces a \emph{self-equivalence} of the fpm $f$ if the following holds:
\begin{enumerate}[(a)]
    \item $f = f\circ g$.
    \item For each $x\in X$, let $g_x: f^{-1}(x) \to f^{-1}(x)$ be the restriction of $g$ to the fiber of $x$. Then ${g_x}_*(\mu_x) = \mu_x$, i.e., $\mu_x(U) = \mu_x(g_x^{-1}(U))$ for any measurable $U\subset f^{-1}(x)$.  
    \item If $U$ is measurable in $f^{-1}(x)$, then $g_x(U)$ is also measurable and $\mu_x(U) = \mu_x(g(U))$.
\end{enumerate}

We want to use the above notion of self-equivalence to define equivalences of fpms. 

\begin{Definition}\label{def:tff}
Two fpms $f: Y \to X$ and $f': Y' \to X$ are \emph{fiberwise equivalent} if there are measurable functions $g: Y \to Y'$ and $h: Y' \to Y$ such that $h\circ g$ and $g\circ h$ are self-equivalences of $f$ and $f'$ respectively. The fiberwise equivalence is denoted by $(g,h)$.
\end{Definition}

The intuition is given by the example that $Y \subset Y'$ such that $\bar{Y} = Y'\backslash Y$ is measurable and $\mu'_x(f'^{-1}(x)\cap \bar{Y}) = 0$ for any $x\in X$. In this case, we view $f$ and $f'$ are giving the exact same statistical information and should not distinguish them. 

We now describe an important construction, the fiber product, which will be used in subsequently in subsequent sections. 

Suppose we have fpms $f: Y \to X$ and $f': Y' \to X$ with fiberwise probability measures $\{\mu_x\}_{x\in X}$ and $\{\mu_x'\}_{x\in X}$ respectively. Inspired by the fiber product in algebraic topology \cite{Hat01}, the \emph{fiber product} of $f$ and $f'$ is defined by \begin{align*}Y\times_{X} Y' = \{(y,y') \in Y \times Y' \mid f(y)=f(y')\}.\end{align*} 
As a subspace of $Y\times Y'$, the fiber product $Y\times_{X} Y'$ carries the subspace $\sigma$-algebra of the product $\sigma$-algebra on $Y\times Y'$. The associated surjective measurable map is $h: Y\times_X Y' \to X, (y,y') \mapsto f(y)(=f(y'))$. For $x\in X$, its fiber $h^{-1}(x)$ is $f^{-1}(x) \times {f'}^{-1}(x)$ and we equip it with measure $\mu_x\times \mu_x'$. The fiber product can be depicted by the following diagram: 
\[ \begin{tikzcd}
Y\times_{X}Y' \arrow{r}{} \arrow[swap]{d}{} & Y \arrow{d}{f} \\%
Y' \arrow{r}{f'}& X.
\end{tikzcd}
\] 
In the diagram, the maps from $Y\times_X Y'$ to $Y$ and $Y'$ are induced by respective projections. The condition $f(y) = f(y')$ is equivalent to the condition that the diagram commutes as functions, i.e., the compositions of maps from $Y \times_X Y'$ to $X$ via any path in the diagram are the same. 

The construction remains valid without $f$ and $f'$ being fpms, though the resulting $Y\times_X Y' \to X$ is in general not an fpm. However, it still satisfies the \emph{universal property} (cf.\ topology pullback) that for any measurable $g: Z \to Y$ and $g': Z \to Y'$ such that $f\circ g = f'\circ g'$, there is a unique measurable map $\bar{g}: Z \to Y\times_X Y'$ whose composition with the project to $Y$ (resp. $Y'$) agrees with $g$ (resp. $g'$). Moreover, if one $f,f'$ is an fpm, say $f$, then the projection $Y\times_X Y' \to Y'$ is an fpm, and the fiberwise measures coincide with those described in the following lemma.
  
\begin{Lemma} \label{lem:cap}
Consider the fiber product (with fpms $f,f'$)
\[ \begin{tikzcd}
Y\times_{X}Y' \arrow{r}{g'} \arrow[swap]{d}{g} & Y \arrow{d}{f} \\%
Y' \arrow{r}{f'}& X.
\end{tikzcd}
\]
\begin{enumerate}[(a)]
    \item For each $y\in Y$ and $y'\in Y'$, then ${g'}^{-1}(y) \cong {f'}^{-1}(f(y))$ and $g^{-1}(y') \cong f^{-1}(f'(y'))$.
    \item The surjection $Y\times_{X}Y' \xrightarrow{g} Y'$ has fiberwise measure $\{\mu_{f'(y')}\}_{y'\in Y'}$. 
    \item The diagram is distributional commutative w.r.t.\ $Y'$ and $Y$, i.e., $f^*\circ f'_*(\mu) ={g'_*} \circ g^*(\mu)$ for any probability measure 
 $\mu$ on $Y'$.
\end{enumerate}
\end{Lemma}

\begin{proof}
For each $y\in Y$, $g'^{-1}(y) = \{(y,y') \mid f(y)=f'(y')\}$. Therefore, the condition for $(y,y')$ holds if and only if $y'\in {f'}^{-1}(f(y))$. As the first component of $g'^{-1}(y)$ is fixed, therefore $g'^{-1}(y) \cong  {f'}^{-1}(f(y))$. Similarly, we have $g^{-1}(y') \cong f^{-1}(f'(y'))$. 

Notice that as $f: Y \to X$ is surjective, so does the map $g: Y\times_X Y' \to Y'$. The fiberwise measure of $f'\circ g: Y\times_X Y' \to X$ is $\mu_x\times \mu_x'$. For $y' \in Y'$, its fiber is $f^{-1}(f'(y'))\times \{y'\}$. Its measure is the marginal of $\mu_{f'(x')}\times \mu_{f'(x')}'$ at $y'$, which is $\mu_{f'(y')}$.

To show $f^*\circ f'_*(\mu) ={g'_*} \circ g^*(\mu)$, consider any test function $\alpha$ on $Y$. We have
\begin{align*}
    \big(f^*\circ f'_*(\mu)\big)(\alpha) & = \int_{y\in Y}\alpha(y) d\big(f^*\circ f'_*\big)(\mu) \\
    & = \int_{x\in X}\int_{y\in f^{-1}(x)}\alpha(y)d\mu_x df'_*(\mu) \\
    & = \int_{y'\in Y'}\int_{y\in f^{-1}(f'(y'))} \alpha(y) d\mu_{f^{-1}(f'(y'))}d\mu \\
    & = \int_{y'\in Y'}\int_{y\in f^{-1}(f'(y'))} \alpha(y)  d\mu_{g^{-1}(y')}d\mu \\
    & = \int_{(y,y')\in Y\times_X Y'} \alpha(y) dg^*(\mu) \\
    & = \int_{y\in Y} \alpha(y) d \big(g'_*\circ g^*\big)(\mu) = (g'_*\circ g^*\big)(\mu)(\alpha).
\end{align*}
The result follows.
\end{proof}

\section{The category of correspondences and its subcategories} \label{sec:tco}

As we have pointed out in the previous section, we use the notion of fpm to model graph and signal information. In this section, we demonstrate how to fit the concept to form a category $\calC_{\mathfrak{c}}$, the \emph{category of correspondences}, in which the morphisms correspond to filters in traditional GSP. 

The objects of $\calC_{\mathfrak{c}}$ are measurable spaces. A \emph{correspondence} $\mathfrak{c} = (f,f')$ from an object $X$ to another $X'$ is given by the following data: an fpm $f: Y \to X$ and a measurable function $f': Y \to X'$ such that for any $\mu \in \calP(X)$, we have \begin{align*}\mathfrak{c}(\mu):= f'_*\circ f^*(\mu)\in \calP(X').\end{align*} We denote it by $(\mathfrak{c},f,f')$ or simply $\mathfrak{c}$ if $f$ and $f'$ are clear from the context. A correspondence is essentially depicted by the diagram $X \xleftarrow{f} Y \xrightarrow{f'} X'$. We call the space $Y$ the \emph{total space} of the correspondence $\mathfrak{c}$. 

For measurable spaces $X,X'$, two correspondences $\mathfrak{c}=(f,f')$ and $\mathfrak{d}=(g,g')$ (with total spaces $Y$ and $Z$ respectively) are equivalent if there is a fiberwise equivalence $(h,h')$ with $h: Y \to Z, h': Z \to Y$ (cf.\ \cref{def:tff}) such that $f' = g'\circ h$ and $g' = f'\circ h'$. The condition can be summarized by the following commutative diagram. 
\[ \begin{tikzcd}
& Y \arrow[swap]{dl}{f} \arrow[swap]{d}{h} \arrow{dr}{f'} & \\
X & Z \arrow[swap]{l}{g} \arrow[swap]{d}{h'} \arrow{r}{g'} & X' \\
& Y \arrow{ul}{f} \arrow[swap]{ur}{f'} & 
\end{tikzcd}
\] 

Given objects $X$ and $X'$ in $\calC_{\mathfrak{c}}$, the morphism $\text{Mor}(X,X')$ consists of correspondences $\mathfrak{c}$ from $X$ to $X'$ up to equivalence. We claim that this notion of morphism makes $\calC_{\mathfrak{c}}$ a category. 

\begin{Theorem} \label{thm:cic}
$\calC_{\mathfrak{c}}$ is a category.
\end{Theorem}

\begin{proof}
    To show that $\calC_{\mathfrak{c}}$ is a category, we need to first describe how to compose morphisms. Suppose we have correspondences $\mathfrak{c_1}=(f_1,f_1') \in \text{Mor}(X_1,X_2)$ and $\mathfrak{c_2}=(f_2,f_2') \in \text{Mor}(X_2,X_3)$, where $f_1$ and $f_2$ are fpms with total spaces $Y_1,Y_2$ respectively. The composition $\mathfrak{c}_2\circ \mathfrak{c}_1$ has total space the fiber product $Z = Y_1\times_{X_2}Y_2$ of $f_1'$ and $f_2$. As shown in the commutative diagram below, the fiber product gives rise to function $g: Z \to Y_1$ and $g': Z \to Y_2$ as projections to the respective components. Consider the composition $h = g\circ f_1: Z \to X_1$. For $x \in X_1$, we have a probability measure $\mu_x$ on $f_1^{-1}(x)$ as $f_1$ is an fpm. On the other hand, by \cref{lem:cap}, $g: h^{-1}(x) \to f_1^{-1}(x)$ is an fpm. Therefore, pulling back via $g$ induces the probability measure $g^*(\mu_x)$ on $h^{-1}(x) = (g\circ f_1)^{-1}(x)$. As a result, $h: Z \to X_1$ is an fpm with total space $Z$ and base space $X_1$. 
\[ \begin{tikzcd}
Z = Y_1\times_{X_2}Y_2 \arrow{r}{g'} \arrow{d}{g} & Y_2 \arrow{r}{f_2'} \arrow{d}{f_2} & X_3 \\%
Y_1 \arrow{r}{f_1'} \arrow{d}{f_1} & X_2\\
X_1.
\end{tikzcd}
\]
Denote the composition $f_2'\circ g'$ by $h'$. Then we claim that $\mathfrak{c} = (h,h')$ is a correspondence and is defined as the composition $\mathfrak{c}_2\circ \mathfrak{c}_1$. We need to verify that for $\mu\in \calP(X_1)$, $\mathfrak{c}(\mu) \in \calP(X_3)$. For this, by \cref{lem:cap}, $\mathfrak{c}(\mu) = \mathfrak{c}_2\big(\mathfrak{c}_1(\mu)\big) \in \calP(X_3)$ as $\mathfrak{c}_1(\mu)$ is in $\calP(X_2)$. 

For any object $X$, the identity map $I: X \to X$ is trivially an fpm as $I^{-1}(x) = \{x\}, x\in X$ is a singleton set. In $\calC_{\mathfrak{c}}$, the identity morphism for $X$ is $\mathfrak{i} = (I,I)$. To see that $\mathfrak{c}\circ \mathfrak{i} = \mathfrak{c}$ (resp. $\mathfrak{i}\circ \mathfrak{c}' = \mathfrak{c}'$) for $\mathfrak{c} \in \text{Mor}(X,X')$ (resp. $\mathfrak{c}' \in \text{Mor}(X',X)$), it suffices to use the fact that: given $f: X'\to X$, we have canonical homeomorphisms $X'\times_X X \cong X\times_X X' \cong X', (x',x)\mapsto (x,x')\mapsto x'$ with $f(x')=x$. 

Lastly, we need to show that $\mathfrak{c}_3\circ (\mathfrak{c}_2\circ \mathfrak{c}_1)= (\mathfrak{c}_3\circ \mathfrak{c}_2)\circ \mathfrak{c}_1$ for $\mathfrak{c}_i=(f_i,f_i') \in \text{Mor}(X_i,X_{i+1})$ with total space $Y_i$. The total space of $\mathfrak{c}_3\circ (\mathfrak{c}_2\circ \mathfrak{c}_1)$ and $(\mathfrak{c}_3\circ \mathfrak{c}_2)\circ \mathfrak{c}_1$ are $(Y_1\times_{X_2}Y_2)\times_{X_3} Y_3$ and $Y_1\times_{X_2}(Y_2\times_{X_3} Y_3)$ respectively, as constructed from repeatedly taking fiber products shown in \figref{fig:comp}. By using the universal property of fiber product, there is $h: (Y_1\times_{X_2}Y_2)\times_{X_3}Y_3 \to Y_1\times_{X_2}(Y_2\times_{X_3} Y_3)$, whose formula is given by $(y_1,y_2,y_3)\mapsto (y_1,y_2,y_3)$. Similarly, there is $h': Y_1\times_{X_2}(Y_2\times_{X_3} Y_3) \to (Y_1\times_{X_2}Y_2)\times_{X_3}Y_3$. The compositions $h\circ h'$ and $h'\circ h$ are both the identity on their respective domains. Therefore $\mathfrak{c}_3\circ (\mathfrak{c}_2\circ \mathfrak{c}_1)$ and $(\mathfrak{c}_3\circ \mathfrak{c}_2)\circ \mathfrak{c}_1$ are equivalent, and hence are the same in $\calC_{\mathfrak{c}}$. 
\begin{figure}
    \centering
    \includegraphics[scale=0.7]{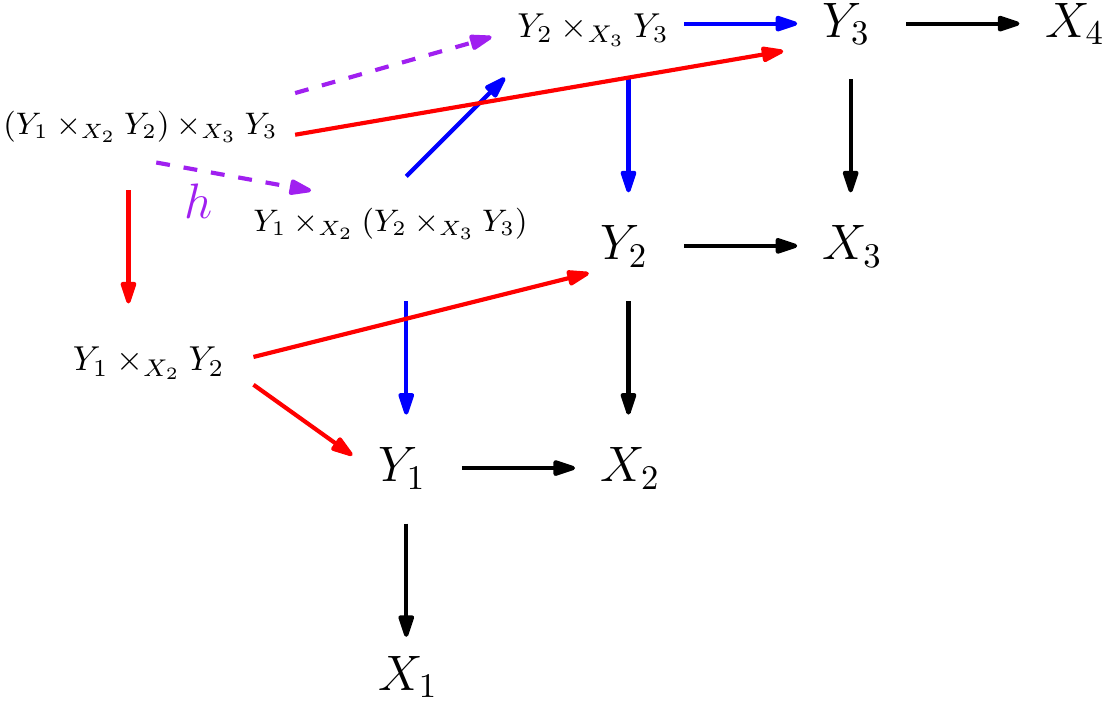}
    \caption{The black arrows are from $\mathfrak{c}_i = (f_i,f_i'), i=1,2,3$. The blue (resp. red) arrows are in the fiber products in forming $Y_1\times_{X_2}(Y_2\times_{X_3} Y_3)$ (resp. $(Y_1\times_{X_2}Y_2)\times_{X_3} Y_3$). A diagram chase, using the universal property of fiber product, we obtain the map $h$.}
    \label{fig:comp}
\end{figure}
\end{proof}

For any collection of measurable spaces $X$, we can define the (full) subcategory $\calC_{c,X}$ of $\calC_{\mathfrak{c}}$ whose objects are in $X$. 

\begin{Lemma} \label{lem:fxx}
For objects $X,X'$ such that $X'$ is a commutative monoid whose binary operation is measurable, the morphism set $\text{Mor}(X,X')$ is an abelian monoid with addition denoted by $+$.
\end{Lemma}

\begin{proof}

Consider $\mathfrak{c}_1 = (f_1,f_1')$ with total space $Y_1$ and $\mathfrak{c}_2 = (f_2,f_2')$ with total space $Y_2$. Then $\mathfrak{c}_1 + \mathfrak{c}_2$ is defined by the following diagram:
\[ \begin{tikzcd}
g_2: Z = Y_1\times_{X}Y_2 \arrow{r} \arrow{d}{g_1} & X'\times X' \arrow{r}{+} & X' \\%
X.
\end{tikzcd}
\]
In the diagram, we have the following description of maps and spaces:
\begin{itemize}
    \item $Z = Y_1\times_{X} Y_2$ is the fiber product w.r.t.\ $f_1,f_2$.
    \item $g_1: Z \to X, (y_1,y_2)\mapsto f_1(y_1)=f_2(y_2) \in X$, and it is an fpm with fiberwise measures $\{\mu_{1,x}\times \mu_{2,x}\}_{x\in X}$ with $\{\mu_{i,x}\}_{x\in X}$ the fiberwise measures for $f_i, i=1,2$.
    \item $g_2: Z \to X'\times X' \to X', (y_1,y_2)\mapsto (f_1'(y_1),f_2'(y_2)) \mapsto f_1'(y_1)+f_2'(y_2)$.
\end{itemize}
It is straightforward to check that $\text{Mor}(X,X')$ with $+$ is a monoid. For example, if $I: X \to X$ is the identity fpm and $0: X \to X'$ the zero map, then $\mathfrak{o}=(I,0)$ is the zero element of $\text{Mor}(X,X')$. It is commutative because $Y_1\times_{X}Y_2 \to Y_2\times_{X}Y_1, (y_1,y_2)\mapsto (y_2,y_1)$ defines a self-equivalence. Associativity is similarly verified by a self-equivalence as in the proof of \cref{thm:cic} (cf. \figref{fig:comp}).  
\end{proof}

\begin{Definition} \label{defn:ici}
    If $\calC'$ is a subcategory of $\calC_{\mathfrak{c}}$, we say that it is closed under addition if for any $X,X' \in \calC'$ with $X'$ an abelian monoid, then $\text{Mor}_{\calC'}(X,X')$ is closed under $+$ in $\text{Mor}(X,X')$.
\end{Definition}

\section{Linear filters}

In this section, we define what we mean by linearity in the categorical setup. The objects we shall consider are finite dimensional vector spaces $\mathbb{R}^n$. The most intuitive form of linear filters are correspondences of the form $\mathfrak{c}=(p,\times)$ where $p: \mathbb{R}^m \times M_{m,n}(\mathbb{R}) \to \mathbb{R}^m$ where $p$ is the projection to the first component, and $\times: \mathbb{R}^m \times M_{m,n}(\mathbb{R}) \to \mathbb{R}^n$ is the multiplication $(r,M)\mapsto Mr$. However, these correspondences are insufficient to form a category as they do not respect compositions, and we need to enlarge the morphism set.

\begin{Definition}\label{defn:acb}
    A correspondence between from $\mathbb{R}^m$ to $\mathbb{R}^n$ is a \emph{linear filter} if it is equivalent to $\mathfrak{c}=(f,f'): \mathbb{R}^{n_1} \to \mathbb{R}^{n_2}$ with total space $X$ of the form: 
\[ \begin{tikzcd}
X \arrow{r}{g} \arrow{d}{f} & \mathbb{R}^{m}\times M_{m,{n_2}}(\mathbb{R}) \arrow{d}{p} \arrow{r}{\times} & \mathbb{R}^{n_2}  \\
\mathbb{R}^{n_1} \arrow{r}{T} & \mathbb{R}^m,
\end{tikzcd}
\]
such that $T$ is a linear transformation and $f'=\times \circ g$, and moreover, the diagram commutes.
\end{Definition}

We claim that the collection $\calC_{\mathfrak{l}}$ of finite dimensional vector spaces and linear filters as defined in \cref{defn:acb} form a subcategory of $\calC_{\mathfrak{c}}$. To show this, we prove a slightly more general result, which can be used to show that other morphism collections form subcategories of $\calC_{\mathfrak{c}}$. 

Assume that we have a subset $\calO$ of the set of measurable spaces, i.e., objects of $\calC_{\mathfrak{c}}$. Consider the following setup:
\begin{Assumption} \label{assu:tia}
\begin{itemize}
    \item There is a collection of measurable functions $\calF$ between objects in $\calO$. 
    \item For each pair of objects $X_1,X_2 \in \calO$, there is a subset $\mathcal{M}_{X_1,X_2}$ of measurable functions from $X_1$ to $X_2$, such that $\calM_{X_1,X_2}$ is a measurable space. Moreover, if $X_1,X_2,X_3 \in \calO$ and $f_1 \in \calM_{X_1,X_2}, f_2 \in \calM_{X_2,X_3}$, then $f_2\circ f_1 \in \calM_{X_1,X_3}$. Moreover, for any $X_1\in \calO$, $\calM_{X_1,X_1}$ contains the identity morphism. This makes $\calO$ a subcategory of the category of measurable spaces. 
    \item For any $f \in \calF\cap \calM_{X_2,X_3}$, $f_1 \in \calM_{X_1,X_2}, f_2 \in \calM_{X_3,X_4}$, the composition $f_2\circ f\circ f_1$ is in $\calM_{X_1,X_4}$.
    \item There is a property $\bP$ of measurable spaces such that the following holds: if $\calM_{X_1,X_2}$ and $\calM_{X_3,X_4}$ have property $\bP$ and $f: X_2\to X_3 \in \calF$, then the map $f_*: \calM_{X_1,X_2} \times \calM_{X_3,X_4} \to \calM_{X_1,X_4}, (f_1,f_2)\mapsto f_2\circ f \circ f_1$ makes $\calM_{X_1,X_4}$ have property $\bP$.
\end{itemize}
\end{Assumption}

In addition, we may also consider the following additional condition regarding $\bP$ in \cref{assu:tia}.

\begin{Assumption} \label{assu:sxx}
If $f_1: X \to X_1$ and $f_2: X \to X_2$ are in $\calF$, then so is $f_1\times f_2: X \to X_1\times X_2$. 

Moreover, suppose $X_1,X_2,X_3$ are in $\calO$ and $X_3$ is an abelian monoid. If $\calM_{X_1,X_3}$ and $\calM_{X_2,X_3}$ have property $\bP$, then there is a morphism $\phi_{X_1,X_2,X_3}: \calM_{X_1,X_3}\times \calM_{X_2,X_3} \to \calM_{X_1\times X_2,X_3}$ such that
\begin{itemize}
    \item $\phi_{X_1,X_2,X_3}$ makes $\calM_{X_1\times X_2,X_3}$ have property $\bP$.
    \item For $f_1 \in \calM_{X_1,X_3}, f_2 \in \calM_{X_2,X_3}$ and $x_1\in X_1, x_2\in X_2$, we have 
   \begin{align*} \phi_{X_1,X_2,X_3}\big((f_1,f_2)\big)\big((x_1,x_2)\big) = f_1(x_1)+f_2(x_2).
   \end{align*}
\end{itemize}

\end{Assumption}

\begin{Lemma} \label{lem:wta}
Under \cref{assu:tia}, let $\calC_{\calO}$ consist of objects $\calO$. A morphism between $X_1,X_2\in \calO$ is any correspondence equivalent to the form $\mathfrak{c} = (h,e\circ g)$ with commuting square:
 \[ \begin{tikzcd}
Y \arrow{r}{g} \arrow{d}{h} & X_2\times \calM_{X_2,X_3} \arrow{d}{p} \arrow{r}{e} & X_3  \\
X_1 \arrow{r}{f} & X_2,
\end{tikzcd}
\]
such that $f\in \calF$, $\calM_{X_2,X_3}$ has property $\bP$, $p$ is the projection, and $e$ is the evaluation map $e(x,l) = l(x)$. Then $\calC_{\calO}$ is a subcategory of $\calC_{\mathfrak{c}}$. 

If in addition \cref{assu:sxx} holds, then $\calC_{\calO}$ is closed under addition in the sense of \cref{defn:ici}. 
\end{Lemma}

\begin{proof}
If $X_1\in \calO$, then $\calC_{\calO}$ contains the identity morphism by choosing $Y = X_2 = X_3 := X_1$ and $h = f = id_{X_1}$, where $id_{X_1}:X_1 \to X_1$ is the identity map.

Suppose we are given diagrams
    \[ \begin{tikzcd}
Y_1 \arrow{r}{g_1} \arrow{d}{h_1} & X_2\times \calM_{X_2,X_3} \arrow{d}{p} \arrow{r}{e} & X_3  \\
X_1 \arrow{r}{f_1} & X_2
\end{tikzcd}
\]
and 
\[ \begin{tikzcd}
Y_2 \arrow{r}{g_2} \arrow{d}{h_2} & X_4\times \calM_{X_4,X_5} \arrow{d}{p} \arrow{r}{e} & X_5  \\
X_3 \arrow{r}{f_2} & X_4. 
\end{tikzcd}
\]
Then their composition in $\calC_{\mathfrak{c}}$ is given by the following diagram
\[ \begin{tikzcd}
Y_1\times_{X_3} Y_2 \arrow{r}{g_3} \arrow{d}{f_3} & X_2 \times \calM_{X_2,X_5} \arrow{d}{p} \arrow{r}{e} & X_5  \\
X_1 \arrow{r}{f_1} & X_2.
\end{tikzcd}
\]
If $g_1(y_1) = (x_2,l_2)$ and $g_2(y_2)=(x_4,l_4)$, then in the above diagram, we have $f_3\big((y_1,y_2)\big) = h_1(y_1)$ and $g_3\big((y_1,y_2)\big) = \big(f_1\circ h_1(y_1), l_4\circ f_2\circ l_2\big)$. By our setup, $\calM_{X_2,X_5}$ has property $\bP$ and the composition of morphisms remains a morphism. Therefore, $\calC_{\calO}$ is a subcategory of $\calC_{\mathfrak{c}}$. 

Suppose \cref{assu:sxx} holds. Suppose we have diagrams for morphisms $\mathfrak{c}_1$ and $\mathfrak{c}_2$ in $\text{Mor}_{\calC_{\calO}}(X_1,X_3)$:
\[ \begin{tikzcd}
Y_1 \arrow{r}{g_1} \arrow{d}{h_1} & X_2\times \calM_{X_2,X_3} \arrow{d}{p} \arrow{r}{e} & X_3  \\
X_1 \arrow{r}{f_1} & X_2
\end{tikzcd}
\]
and 
\[ \begin{tikzcd}
Y_2 \arrow{r}{g_2} \arrow{d}{h_2} & X_4\times \calM_{X_4,X_3} \arrow{d}{p} \arrow{r}{e} & X_3  \\
X_1 \arrow{r}{f_2} & X_4. 
\end{tikzcd}
\]
We have the following diagram for addition with $X_{2,4} = X_2\times X_4$, $\calM = \calM_{X_2,X_3}\times \calM_{X_4,X_3}$ and $\phi = \phi_{X_2,X_4,X_3}$:
\[ \begin{tikzcd}
Y_1\times_{X_1}Y_2 \arrow{r} \arrow{d} & X_{2,4}\times \calM \arrow{r}{id\times \phi}& X_{2,4}\times \calM_{X_{2,4},X_3}\arrow{d}{p} \arrow{r}{e} & X_3  \\
X_1 \arrow{r}{=} & X_1\arrow{r}{f_1\times f_2} & X_{2,4}. &
\end{tikzcd}
\]
By \cref{assu:sxx}, $f_1\times f_2\in \calF$ and $\calM_{X_{2,4},X_3}$ has property $\bP$, and therefore $\mathfrak{c}_1+\mathfrak{c}_2$ is in $\text{Mor}_{\calC_{\calO}}(X_1,X_3)$ as claimed.
\end{proof}

Let $\calO$ consist of finite dimensional real vector spaces and $\calF$ be the set of linear transformations. For two vector spaces $X_1 = \mathbb{R}^{n_1}$ and $X_2 = \mathbb{R}^{n_2}$, the space $\calM_{X_1,X_2}$ is the space of $n_1$-by-$n_2$ matrices $M_{n_1,n_2}(\mathbb{R})$. Moreover, the map $M_{n_1,n_3}(\mathbb{R})\times M_{n_2,n_3}(\mathbb{R}) \to M_{n_1+n_2,n_3}(\mathbb{R})$ is given by $(M_1,M_2)\mapsto M$ such that $M\big(x_1,x_2\big) = M_1x_1+M_2x_2$. The property $\bP$ is vacuous and by applying \cref{lem:wta} we obtain the following. 

\begin{Corollary} \label{thm:tcc}
$\calC_{\mathfrak{l}}$ is a subcategory of $\calC_{\mathfrak{c}}$ and is closed under addition.
\end{Corollary}

For the remaining of this section, we give a few more examples. In $\calC_{\mathfrak{l}}$, we may consider linear filters $\mathfrak{c}=(f,f')$ of (or equivalent to) the following form
\[ \begin{tikzcd}
X \arrow{r}{g} \arrow{d}{f} & \mathbb{R}^{n_1}\times M_{{n_1},{n_2}}(\mathbb{R}) \arrow{d}{p} \arrow{r}{\times} & \mathbb{R}^{n_2}  \\
\mathbb{R}^{n_1} \arrow{r}{id} & \mathbb{R}^{n_1},
\end{tikzcd}
\]
such that the following holds:
\begin{itemize}
    \item The map $id$ is the identity transformation. As a consequence, the fpm $f$ induces an fpm $p: \mathbb{R}^{n_1}\times M_{{n_1},{n_2}}(\mathbb{R}) \to \mathbb{R}^{n_1}$ such that the fiberwise measure $\nu_r, r\in \mathbb{R}^{n_1}$ on $\{r\}\times M_{{n_1},{n_2}} \cong M_{{n_1},{n_2}}$ is given by $g_*(\mu_r)$, where $\mu_r$ is the fiberwise measure on $f^{-1}(r) \subset X$. 
    \item There is a probability measure $\nu$ on $M_{{n_1},{n_2}}(\mathbb{R})$ such that $\nu_r=\nu$ for every $r\in \mathbb{R}^{n_1}$. 
\end{itemize}
The condition essentially says that the fiberwise measures are constant. This is exactly the setup of \cite{Ji22}. 
Let the collection of such morphisms form $\calC_{\mathfrak{t}}$, where $\mathfrak{t}$ stands for ``trivial''. 

In the language of \cref{assu:tia}, $\calF$ consists of identity maps. The property $\bP$ is that the induced fiberwise measure $\mathbb{R}^{n_1}\times M_{{n_1},{n_2}}(\mathbb{R}) \to \mathbb{R}^{n_1}$ is constant. Suppose $\mathbb{R}^{n_1}\times M_{{n_1},{n_2}}(\mathbb{R}) \to \mathbb{R}^{n_1}$ has constant fiberwise measure $\mu_1$ and $\mathbb{R}^{n_2}\times M_{{n_2},{n_3}}(\mathbb{R}) \to \mathbb{R}^{n_2}$ has constant fiberwise measure $\mu_2$. Then they induce $\mathbb{R}^{n_1}\times M_{{n_1},{n_3}}(\mathbb{R}) \to \mathbb{R}^{n_3}$ with the constant fiberwise measure $\rho_*(\mu_1\times \mu_2)$ induced by $\rho: M_{{n_1},{n_2}}(\mathbb{R})\times M_{{n_2},{n_3}}(\mathbb{R}) \to M_{{n_1},{n_3}}(\mathbb{R}), (M_1,M_2)\mapsto M_2M_1$. Similarly, if we have $\mu_1$ and $\mu_2$ on $M_{n_1,n_2}(\mathbb{R})$, then they induce the measure $\sigma_*(\mu_1\times \mu_2)$ on $M_{n_1,n_2}(\mathbb{R})$ by $\sigma: M_{n_1,n_2}(\mathbb{R})\times M_{n_1,n_2}(\mathbb{R}) \to M_{n_1,n_2}(\mathbb{R}), (M_1,M_2) \mapsto M_1+M_2$. 

\begin{Corollary}
    $\calC_{\mathfrak{t}}$ is a subcategory of $\calC_{\mathfrak{l}}$ and is closed under addition. 
\end{Corollary}

Let $\text{Vect}(\mathbb{R})$ be the category of finite dimensional $\mathbb{R}$ vectors. It embeds in $\calC_{\mathfrak{t}}$ via a functor $\iota$ as follows. For $\mathbb{R}^n \in \text{Vect}(\mathbb{R})$, $\iota(\mathbb{R}^n) = \mathbb{R}^n$. If $T: \mathbb{R}^{n_1} \to \mathbb{R}^{n_2}$ is a linear transformation, then $\iota(T)$ is the following morphism
\[ \begin{tikzcd}
\mathbb{R}^{n_1}\times M_{{n_1},{n_2}}(\mathbb{R}) \arrow{r}{id} \arrow{d}{p} & \mathbb{R}^{n_1}\times M_{{n_1},{n_2}}(\mathbb{R}) \arrow{d}{p} \arrow{r}{\times} & \mathbb{R}^{n_2}  \\
\mathbb{R}^{n_1} \arrow{r}{id} & \mathbb{R}^{n_1},
\end{tikzcd}
\]
where the fpm $p$ has the constant fiberwise measure supported on the single transformation $T$.

We consider the following variant $\calC_{\mathfrak{f}}$ of $\calC_{\mathfrak{t}}$, where $\mathfrak{f}$ stands for ``finite'', ``fiber bundle''. As a spoiler, this corresponds to locally constant SAGS in \cite{Ji23d}. We make the following (sole) modification to the definition $\calC_{\mathfrak{t}}$, which is essentially a change to the property $\bP$. For the projection $p: \mathbb{R}^{n_1}\times M_{{n_1},{n_2}}(\mathbb{R}) \to \mathbb{R}^{n_1}$, the induced measure fiberwise measures $\{\mu_r,r\in \mathbb{R}^{n_1}\}$ satisfy: there is a set $K$ of (Lebesgue) measure $0$ such that for each $r\notin K$, $\mu_r$ is supported on the set of invertible matrices and there is an open subset $U_r$ of $r$ with $\mu_{r'} = \mu_r$ for $r'\in U_r$. The condition essentially requires that $\{\mu_r,r \in \mathbb{R}^{n_1}\}$ is locally constant. 

To apply \cref{lem:fxx}, we claim that \begin{align*}\rho: M_{{n_1},{n_2}}(\mathbb{R})\times M_{{n_2},{n_3}}(\mathbb{R}) \to M_{{n_1},{n_3}}(\mathbb{R}), (M_1,M_2)\mapsto M_2M_1\end{align*} induces locally constant fiberwise measures (in the above sense) on $M_{n_1,n_3}(\mathbb{R})$. To see this, let the measure $0$ subsets of $M_{{n_1},{n_2}}(\mathbb{R})$ and $M_{{n_2},{n_3}}(\mathbb{R})$ be $K_1$ and $K_2$ respectively. 
For $r\notin K_1$, let the support of $\mu_r$ be the finite set of invertible matrices $S_r = \{G_{r,1},\ldots, G_{r,k}\}$. Then $K_{r,1}=\cap_{1\leq i\leq k}G_{r,i}^{-1}(K_2)\cap U_r$ has measure zero. A countable union of $U_r$ covers $\mathbb{R}^{n_1}\backslash K_1$ and the corresponding union $K_1'$ of $K_{r,1}$ has measure zero. As $S_r$ is a finite set of invertible matrices, for $r' \in U_r\backslash K_{r,1}$, we can always find a small open neighborhood $W_{r'}$ of $r'$ such that the fiberwise measures on $p: \mathbb{R}^{n_2}\times M_{n_2,n_3}(\mathbb{R})\to \mathbb{R}^{n_2}$ is constant on $G_{r,i}(W_{r'})$ for every $1\leq i\leq k$. Therefore, as in the argument for $\calC_{\mathfrak{t}}$, $\rho$ induces a constant measure on $W_{r'}$. This holds true outside $K_1'$ and the claim is proved.

For the addition, the argument is simpler as we notice that if two sets of fiberwise measures are constant on open sets $U_1$ and $U_2$, then their sum is constant on $U_1\cap U_2$. We omit the details and the following holds. 

\begin{Corollary}
    $\calC_{\mathfrak{f}}$ is a subcategory of $\calC_{\mathfrak{c}}$ and is closed under addition.
\end{Corollary}

\section{Conditional expectation} \label{sec:ce}

In this section, we discuss the relationship between our framework and condition expectation. 

Suppose we are given $\mathfrak{c} = (f,f')$ with $f: Y \to \mathbb{R}^{n_1}$ and $f': Y \to \mathbb{R}^{n_2}$. The former is an fpm with fiberwise measures $\{\mu_r\}_{r\in \mathbb{R}^{n_1}}$. It induces a map $\mathfrak{c}: \calP(\mathbb{R}^{n_1}) \to \calP(\mathbb{R}^{n_2})$. We want to find a map $\mathbb{R}^{n_1} \to \mathbb{R}^{n_2}$ that is a good approximation of $\mathfrak{c}$. 

To construct $e_{\mathfrak{c}}$, for $r \in \mathbb{R}^{n_1}$, define 
\begin{align*}
e_{\mathfrak{c}}(r) = \int_{y\in f^{-1}(r)}f'(y) \ud\mu_r = \int_{r'\in \mathbb{R}^{n_2}} r'\ud \mathfrak{c}(\delta_r).
\end{align*}
It is related to conditional expectation as follows. Recall that for any $\mu \in \mathcal{P}(\mathbb{R}^{n_1})$, it pulls back to a distribution $f^*(\mu)$ on $Y$. In this respect, both $f$ and $f'$ can be viewed as random variables on the sample space $Y$. It is well known that there is a condition expectation $e_{f'}: \mathbb{R}^{n_1} \to \mathbb{R}^{n_2}$ such that $e_{f'}(r) = e_{\mathfrak{c}}(r)$ up to a set with $\mu$ measure $0$. Due to this fact, the promised approximation property of $e_{\mathfrak{c}}$ reads as follows.

\begin{Theorem} \label{thm:ftf}
For the fpm $\mathfrak{c}=(f,f')$, the function $e_{\mathfrak{c}}$ is measurable and $e_{\mathfrak{c},*}:\mathcal{P}(\mathbb{R}^{n_1}) \to \mathcal{P}(\mathbb{R}^{n_2})$ is well-defined. Moreover, for any measurable $g: \mathbb{R}^{n_1} \to \mathbb{R}^{n_2}$ and subset $S \subset \mathbb{R}^{n_1}$, the following holds:
\begin{align*}
   \sup_{\text{supp}(\mu) \subset S} W\big(e_{\mathfrak{c},*}(\mu), \mathfrak{c}(\mu)\big) \leq \sup_{\text{supp}(\nu) \subset S} W\big(g_*(\nu), \mathfrak{c}(\nu)\big),
\end{align*}
where $W$ is the $2$-Wasserstein metric and the supreme is taken over $\mu$ (resp.\ $\nu$) in $\calP(\mathbb{R}^{n_1})$ supported in $S$. 
\end{Theorem}

\begin{proof}
We first show that $e_{\mathfrak{c}}$ is measurable. Let $C$ be any compact subset of $\mathbb{R}^{n_1}$ and $\mu_C$ be the uniform distribution on $C$. It induces the pullback measure $f^*(\mu_C)$ on $Y$. Moreover, it is easy to verify that $f^*f_*(\mu_C) = \mu_C$. 

We view $Y$ as the sample space with probability distribution $f^*(\mu_C)$ and $f$ and $f'$ as random variables. Let $e_C: \mathbb{R}^{n_1} \to \mathbb{R}^{n_2}$ be the associated conditional expectation. By the construction, we have $e_C = e_{\mathfrak{c}}$ on $C$ and $e_C = 0$ on the complement $\mathbb{R}^n\backslash C$. Moreover, $e_C$ is measurable w.r.t.\ the measure $\mu_C$. However, as $\mu_C$ is uniform, $e_C$ is also measurable w.r.t.\ the Lebesgue measure.  

Let $C_1 \subset C_2 \subset \ldots \subset C_i \subset \ldots$ be a sequence of compact subsets of $\mathbb{R}^{n_1}$ such that $\cup_{i\geq 1}C_i = \mathbb{R}^{n_1}$. Then $(e_{C_i})_{i\geq 1}$ is a sequence of measurable functions whose pointwise limit is $e_{\mathfrak{c}}$. Therefore, $e_{\mathfrak{c}}$ is also measurable. 

As a consequence, given any distribution $\mu$ on $\mathcal{P}(\mathbb{R}^{n_1})$, pushforward induces $e_{\mathfrak{c},*}(\mu)$. We need to show that $e_{\mathfrak{c},*}(\mu) \in \calP(\mathbb{R}^{n_2})$ in order to claim that $e_{\mathfrak{c},*}$ is well defined as a map $\mathcal{P}(\mathbb{R}^{n_1}) \to \mathcal{P}(\mathbb{R}^{n_2})$. For the finiteness of the mean, it suffices to use linearity. While for the finiteness of the $2$nd moment, consider the Jensen inequality:
\begin{align} \label{eq:nmx}
    \norm{\mathbb{E}_{x\in f^{-1}(r)\sim \mu_r}f'(x)}^2\leq \mathbb{E}_{x\in f^{-1}(r)\sim \mu_r}\norm{f'(x)}^2, r\in \mathbb{R}^{n_2}.
\end{align}
For any $\mu \in \mathcal{P}(\mathbb{R}^{n_1})$, to show that $e_{\mathfrak{c},*}(\mu) \in \mathcal{P}(\mathbb{R}^{n_2})$, it suffices to check that 
\begin{align*}
    \int_{r\in \mathbb{R}^{n_1}} \norm{\mathbb{E}_{x\in f^{-1}(r)\sim \mu_r}f'(x)}^2 d\mu < \infty.
\end{align*}
However, by (\ref{eq:nmx}), the left-hand side is bounded by 
\begin{align*}
    \int_{r\in \mathbb{R}^{n_1}} \mathbb{E}_{x\in f^{-1}(r)\sim \mu_r}\norm{f'(x)}^2 d\mu < \infty,
\end{align*}
due to the assumption that $\mathfrak{c}(\mu) \in \mathcal{P}(\mathbb{R}^{n_2})$.

To show the claimed inequality, 
let $S$ be a subset $\mathbb{R}^{n_1}$ and $\mu \in \calP(\mathbb{R}^{n_1})$ be supported on $S$. Consider $\phi$ the pushforward measure of $f^*(\mu)$ on $\mathbb{R}^{n_2}\times \mathbb{R}^{n_2}$ via the map: $Y\to \mathbb{R}^{n_2}\times \mathbb{R}^{n_2}, y\mapsto (e_{\mathfrak{c}}\circ f(y),f'(y))$. The marginals of $\phi$ are 
\begin{align*}
(e_{\mathfrak{c}}\circ f)_*(f^*(\mu)) = e_{\mathfrak{c},*}(\mu) \text{ and } {f'}_*(f^*(\mu)) = \mathfrak{c}(\mu)
\end{align*}
respectively. Therefore,   
\begin{align*}
    & W(e_{\mathfrak{c},*}(\mu), \mathfrak{c}(\mu))^2 \\
    \leq & \int_{(r_1,r_2)\in \mathbb{R}^{n_2}\times \mathbb{R}^{n_2}} \norm{r_1-r_2}^2 \ud\phi \\
    = & \int_{y\in Y}\norm{e_{\mathfrak{c}}\circ f(y)-f'(y)}^2 \ud f^*(\mu) \\
    = & \mathbb{E}_{y \sim f^*(\mu)}\norm{e_{\mathfrak{c}}\circ f(y)-f'(y)}^2 \\
    \leq & \mathbb{E}_{y \sim f^*(\mu)}\norm{g\circ f(y)-f'(y)}^2,
\end{align*}
The last inequality holds as $e_{\mathfrak{c}}$ is the conditional expectation (up to a set of measure $0$) w.r.t.\ $f^*(\mu)$ on $Y$. 

To estimate the right-hand-side, we have 
\begin{align*}
\mathbb{E}_{y \sim f^*(\mu)}\norm{g\circ f(y)-f'(y)}^2 \leq \sup_{r\in S} \mathbb{E}_{y\sim f^*(\delta_r)}\norm{g\circ f(y)-f'(y)}^2.
\end{align*}
As for any $y \in f^{-1}$, $g\circ f(y)$ is the constant $g(r)$. Therefore, 
\begin{align*}
\mathbb{E}_{y\sim f^*(\delta_r)}\norm{g\circ f(y)-f'(y)}^2 = W\big(g_*(\delta_r),\mathfrak{c}(\delta_r)\big)^2 \leq \sup_{\text{supp}(\nu) \subset S} W\big(g_*(\nu), \mathfrak{c}(\nu)\big)^2.
\end{align*}
The result follows.
\end{proof}

The result essentially claims that $e_{\mathfrak{c}}$ is the best function approximation of $\mathfrak{c}$ in the $\infty$-norm. The map $e_{\mathfrak{c}}$ also enjoys some interesting algebraic properties. 

\begin{Lemma} \label{lem:whe}
We have $e_{\mathfrak{c}+\mathfrak{c}'} = e_{\mathfrak{c}}+e_{\mathfrak{c}'}$ and $e_{\mathfrak{c}\circ\mathfrak{c}'}(r_1) = \mathbb{E}_{r_2\sim \mathfrak{c}'(\delta_{r_1})} e_{\mathfrak{c}}(r_2)$, for morphisms $\mathfrak{c}, \mathfrak{c'}$ in $\calC_{\mathfrak{l}}$ whenever the respective binary operation is well-defined.
\end{Lemma}

\begin{proof}
To verify, assume both $\mathfrak{c},\mathfrak{c}'$ are morphisms from $\mathbb{R}^{n_1}$ to $\mathbb{R}^{n_2}$, we compute
\begin{align*}
     & e_{\mathfrak{c}+\mathfrak{c}'}(r)
    = \int_{r_2\in \mathbb{R}^{n_2}} r_2 \ud (\mathfrak{c}+\mathfrak{c}')(\delta_{r_1}) \\
    = & \int_{r_2 \in \mathbb{R}^{n_2}}\int_{r_2'\in \mathbb{R}^{n_2}}r_2+r_2' \ud\mathfrak{c}'(\delta_{r_1})\ud \mathfrak{c}(\delta_{r_1}) \\
    = &\int_{r_2\in \mathbb{R}^{n_2}}r_2+e_{\mathfrak{c}'}(r_1)\ud\mathfrak{c}(\delta_{r_1})\\
    = & e_{\mathfrak{c}'}(r_1) + \int_{r_2\in\mathbb{R}^{n_2}}r_2\ud\mathfrak{c}(\delta_{r_1}) = e_{\mathfrak{c}}(r_1)+e_{\mathfrak{c}'}(r_1).
\end{align*}
For the composition, consider $\mathfrak{c}' \in \text{Mor}_{\calC_{\mathfrak{l}}}(\mathbb{R}^{n_1},\mathbb{R}^{n_2})$ and $\mathfrak{c} \in \text{Mor}_{\calC_{\mathfrak{l}}}(\mathbb{R}^{n_2},\mathbb{R}^{n_3})$, we have
\begin{align*}
    & e_{\mathfrak{c}\circ\mathfrak{c}'}(r_1) = \int_{r_3\in \mathbb{R}^{n_3}}r_3 \ud\mathfrak{c}\circ\mathfrak{c}'(\delta_{r_1}) \\
    = & \int_{r_2\in \mathbb{R}^{n_2}}\int_{r_3\in \mathbb{R}^{n_3}} r_3 \ud\mathfrak{c}(\delta_{r_2}) \ud\mathfrak{c}'(\delta_{r_1})\\
    = & \int_{r_2\in \mathbb{R}^{n_2}}e_{\mathfrak{c}}(r_2)\ud\mathfrak{c}'(\delta_{r_1}) = \mathbb{E}_{r_2\sim \mathfrak{c}'(\delta_{r_1})} e_{\mathfrak{c}}(r_2).
\end{align*}
\end{proof}

\begin{Corollary} \label{coro:tai}
The assignment $\iota^*: \mathfrak{c}=(f,f') \mapsto e_{\mathfrak{c}}$ induces a functor from $\calC_{\mathfrak{t}}$ to $\text{Vect}(\mathbb{R})$ as a right inverse to the embedding $\iota$. 
\end{Corollary}

\begin{proof}
 It suffices to show that $\iota^*(\mathfrak{c})\circ \iota^*(\mathfrak{c}') = \iota^*(\mathfrak{c}\circ \mathfrak{c}')$ for composable morphisms $\mathfrak{c}, \mathfrak{c}'$ in $\calC_{\mathfrak{t}}$. This requires that $e_{\mathfrak{c}\circ \mathfrak{c}'} = e_{\mathfrak{c}}\circ e_{\mathfrak{c}'}$. As $\mathfrak{c}$ is a morphism in $\calC_{\mathfrak{t}}$, $e_{\mathfrak{c}}$ is a linear transformation. Therefore by \cref{lem:whe}, we have 
 \begin{align*}
      e_{\mathfrak{c}\circ\mathfrak{c}'}(r_1)  = \mathbb{E}_{r_2\sim \mathfrak{c}'(\delta_{r_1})} e_{\mathfrak{c}}(r_2) = e_{\mathfrak{c}}(\mathbb{E}_{r_2\sim \mathfrak{c}'(\delta_{r_1})} r_2) = e_{\mathfrak{c}}\circ e_{\mathfrak{c}'}(r_1).
 \end{align*}
 The result follows.
\end{proof}

\section{Graph signal processing concepts}

In this section, we use the language of the paper to introduce concepts that generalize their counterparts in the traditional graph signal processing theory. Recall that a morphism $\mathfrak{c}$ consists of a pair of maps $(f,f')$, where $f$ is an fpm. The fpm $f$ usually encodes graph or topological information of the data, and it is usually known a priori. It is the map $f'$ that accounts for the ``transformation'' between objects.  

\subsubsection*{Change of basis}
The change of basis intends to generalize Fourier transform. Let $O_n(\mathbb{R})$ be the (group of) $n\times n$ orthogonal matrices. A morphism $\mathfrak{c}$ in $\calC_{l}$ is a \emph{change of basis} if it take the following form:
\[ \begin{tikzcd}
X \arrow{r}{g} \arrow{d}{f} & \mathbb{R}^n\times M_n(\mathbb{R}) \arrow{d}{p} \arrow{r}{\times} & \mathbb{R}^n  \\
\mathbb{R}^n \arrow{r}{id} & \mathbb{R}^n.
\end{tikzcd}
\]
such that the induced fpm $p: \mathbb{R}^n \times M_n(\mathbb{R}) \to \mathbb{R}^n$ has fiberwise measures supported on $O_n(\mathbb{R})$. 

For example, for a fixed graph $G$, let $X = \mathbb{R}^n \times \{G\}$ and $g: X \to \mathbb{R}^n\times M_n(\mathbb{R})$, $(r,G) \mapsto (r, U_G)$, where $L_G= U_GD_GU_G^T$ is a fixed orthogonal decomposition of the Laplacian $L_G$. This is nothing but the Fourier transform in GSP. 

Its pseudo-inverse $\bar{\mathfrak{c}}$ is defined by $\bar{\mathfrak{c}}=(f,\bar{f'})$ where $\bar{f'}$ is the composition $ \times\circ \top \circ g$ where $\top: \mathbb{R}^n \times M_n(\mathbb{R}) \to \mathbb{R}^n \times M_n(\mathbb{R}), (r,M)\mapsto (r, M^{\top})$. It clearly generalizes the notion of inverse graph Fourier transform in GSP. We remark that the change of basis and its pseudo-inverse can be generalized in an obvious way if $\mathbb{R}$ is replaced by $\mathbb{C}$. 

\subsubsection*{Convolution} 
Let $D_n(\mathbb{R})$ be space of $n\times n$ diagonal matrices. A \emph{convolution kernel} is a morphism $\mathfrak{k}$ of the form $(k,k')$
\[ \begin{tikzcd}
X \arrow{r}{g} \arrow{d}{k} & \mathbb{R}^n\times M_n(\mathbb{R}) \arrow{d}{p} \arrow{r}{\times} & \mathbb{R}^n  \\
\mathbb{R}^n \arrow{r}{id} & \mathbb{R}^n.
\end{tikzcd}
\]
such that the induced fpm $p: \mathbb{R}^n \times M_n(\mathbb{R}) \to \mathbb{R}^n$ has fiberwise measures supported on $D_n(\mathbb{R})$. 

For a change of basis $\mathfrak{c}$ (and its pseudo-inverse $\bar{\mathfrak{c}}$) as in the previous example, then the associated convolution filter with kernel $\mathfrak{k}$ is the composition $ \bar{\mathfrak{c}}\circ \mathfrak{k}\circ \mathfrak{c}$. In \cite{Ji22}, $\mathfrak{c}$ and $\mathfrak{k}$ are morphisms in the subcategory $\calC_{\mathfrak{t}}$. The Fourier transform and convolution defined in \cite{Ji22} are $e_{\mathfrak{c}}$ and $e_{\bar{\mathfrak{c}}\circ \mathfrak{k}\circ \mathfrak{c}}$ introduced in \cref{sec:ce} respectively. They are hence justified by \cref{thm:ftf}.

\subsubsection*{Sampling} 
Let $P_n(\mathbb{R})$ be the set of $n\times n$ projection matrices, i.e., if $P \in P_n(\mathbb{R})$, then $P\circ P = P$. A morphism $\mathfrak{c}$ in $\calC_{l}$ is called \emph{sampling} if it take the following form:
\[ \begin{tikzcd}
X \arrow{r}{g} \arrow{d}{f} & \mathbb{R}^n\times M_n(\mathbb{R}) \arrow{d}{p} \arrow{r}{\times} & \mathbb{R}^n  \\
\mathbb{R}^n \arrow{r}{id} & \mathbb{R}^n.
\end{tikzcd}
\]
such that the induced fpm $p: \mathbb{R}^n \times M_n(\mathbb{R}) \to \mathbb{R}^n$ has fiberwise measures supported on $P_n(\mathbb{R})$. For any $\epsilon \geq 0$, we may define the \emph{$(\mathfrak{c},\epsilon)$-bandlimited signals} as $\{\mu\in \calP(\mathbb{R}^n) \mid W(\mu, \mathfrak{c}(\mu)) \leq \epsilon\}$. For recovery, we find a subset of $m$ coordinates of $\mathbb{R}^n$ and estimate $\mu$ given its observations at those prescribed coordinates. The primary examples of sampling are convolutions associated with diagonal matrices with only $0$ and $1$ as diagonal entries.  

GSP considers the cases when the fpm $p$ is supported on eigenspaces of graph Laplacians. If $\mathfrak{c}$ belongs to $M_n(\mathbb{R})$, then we have the setup of sampling in the traditional GSP, and the recovery amounts to find a pseudo-inverse mathematically. \cite{Ji22} studies the case when $\mathfrak{c}$ belongs to $\calC_{\mathfrak{t}}$. Signal recovery is to find a pseudo-inverse of the matrix $e_{\mathfrak{c}} \in M_n(\mathbb{R})$ (cf.\ \cref{coro:tai}).

\bibliographystyle{plain}
\bibliography{allref}

\end{document}